\documentclass[runningheads]{llncs}

\usepackage[british]{babel}
\usepackage[T1]{fontenc}
\usepackage[utf8]{inputenc}
\usepackage{lmodern}
\usepackage[english]{isodate}

\usepackage[ruled,norelsize]{algorithm2e}

\usepackage{amsmath}
\usepackage{amssymb}
\usepackage{bm}

\DeclareMathOperator{\lineartakum}{\overline{\tau}}

\DeclareMathOperator{\integer}{uint}

\usepackage{graphicx}
\usepackage{pgf, tikz}
\usepackage{pgfplots}
\pgfplotsset{compat=1.17}
\usepgfplotslibrary{statistics}
\usepackage{placeins}

\usepackage{subfig}
\usepackage{pgfplots}

\usepackage{nameref}
\usepackage[hyphens]{url}
\usepackage[hidelinks, breaklinks, unicode]{hyperref} 

\usepackage[style=numeric,minbibnames=1,maxbibnames=2]{biblatex}
\usepackage[autostyle=true]{csquotes}
\usepackage{listingsutf8}
\usepackage{orcidlink}
\usepackage{siunitx}
\usepackage{xcolor}

\definecolor{sign}{HTML}{b02a2d}
\definecolor{direction}{HTML}{007900}
\definecolor{regime}{HTML}{8c399e}
\definecolor{characteristic}{HTML}{1f5dc2}
\definecolor{mantissa}{HTML}{636363}
\definecolor{error}{HTML}{BD002A}
\definecolor{cellbg}{HTML}{EDEDED}

\definecolor{p-sign}{HTML}{FF5454}
\definecolor{p-regime}{HTML}{CC9966}
\definecolor{p-regime-term}{HTML}{996633}
\definecolor{p-exponent}{HTML}{0080FF}
\definecolor{p-fraction}{HTML}{000000}

\makeatletter
\def\lst@makecaption{%
  \def\@captype{table}%
  \@makecaption
}
\makeatother

\begin{document}

\title{Integer Representations in IEEE 754, OFP8, Bfloat16, Posit, and Takum Arithmetics}
\titlerunning{Integer Representations in Arithmetics}
\author{Laslo Hunhold\,\orcidlink{0000-0001-8059-0298}}
\authorrunning{L. Hunhold}
\institute{%
	Parallel and Distributed Systems Group\\
	University of Cologne, Cologne, Germany\\
	\email{hunhold@uni-koeln.de}
}
\maketitle

\begin{abstract}
Although not primarily designed for this purpose, floating-point numbers are often used to represent integral values, with some applications explicitly relying on this capability. However, the integral representation properties of IEEE 754 floating-point numbers have not yet been formally investigated. Recently, the \texttt{bfloat16}, posit and takum machine number formats have been proposed as alternatives to IEEE 754, while OCP 8-bit floating point (OFP8) types (E4M3 and E5M2) have been introduced as 8-bit extensions of IEEE 754, albeit with slight deviations from the standard. It is therefore timely to evaluate IEEE 754 and to assess how effectively the new formats fulfil this function in comparison with the standard they aim to replace.
\par
This paper presents the first rigorous derivations and proofs of the integral representation capabilities of IEEE 754 floating-point numbers, OFP8, \texttt{bfloat16}, posits, and takums. We examine both the exact number of bits required to represent a given integer and the largest consecutive integer representable with a specified bit width.
The results show that OFP8 yields mixed outcomes, \texttt{bfloat16} generally underperforms, and posits fail to scale effectively, whereas takums consistently match or outperform the other formats, maintaining backward compatibility with IEEE 754.
\end{abstract}


\section{Introduction}
Like using a wrench as hammer, floating-point numbers often find themselves used for the purpose of representing integers in various
applications. For instance, JavaScript uses double-precision floating-point numbers for all numerical representations, defining the second-largest consecutively representable integer as the constant \texttt{Number.MAX\_SAFE\_INTEGER} \cite[Sections~6.1.6.1 and 21.1.2.6]{example-javascript}. Similarly, the NIfTI image format, a standard in neuroimaging and MRI, utilises single-precision floating-point numbers to specify voxel offsets \cite{example-nifti1}. This characteristic also has implications for discrete Fourier transforms (representing discrete frequencies), statistical analyses (discrete distributions), deep learning (quantisation) and other contexts where integral and continuous quantities coexist.
\par
The integer representation limits of a given floating-point format are, therefore, of considerable interest. One metric used to evaluate this is the \emph{largest consecutive integer}, defined as the smallest exactly representable positive integer $m \in \mathbb{N}_0$ such that $m+1$ is no longer exactly representable. Notably, this usually results in $m+1$ rounding down to $m$, which poses challenges in scenarios where floating-point numbers are employed as indices, potentially causing loops to never terminate. This further underscores the importance of understanding the integral representation limits of a given format.
\par
With the advent of new machine number formats, it is important to assess their ability to represent integers and determine whether they can serve as drop-in replacements, offering at least the same representational power as IEEE 754 floating-point numbers. Besides three formats within the IEEE 754 framework (OFP8 (E4M3 and E5M2) and \texttt{bfloat16}), this paper examines the two tapered precision posit and linear takum formats, where the exponent is variable-length encoded. Consequently, the number of fraction bits varies depending on the value of the exponent, rendering the analysis significantly more complex than the straightforward proof for IEEE 754 floating-point numbers provided later. None of the five aformentioned formats have been formally analysed in this context;
the existing formulae for posits are empirically derived.
\par
This paper makes three primary contributions: (1) it formally derives the exact number of bits required to represent any integer in posit and takum arithmetic; (2) it formally derives the largest consecutive integer for IEEE 754, OFP8, \texttt{bfloat16}, posit and takum floating-point formats; and (3) it compares the integer representation capabilities of posits and takums with those of IEEE 754, OFP8, and \texttt{bfloat16} floating-point formats.
\par
The remainder of this paper is organised as follows: Sections~\ref{sec:posit_encoding} and \ref{sec:takum_encoding} define the posit and takum formats, respectively. Section~\ref{sec:integer_representations} presents the main results. Section~\ref{sec:evaluation} evaluates these results and compares the formats at various precisions, followed by the conclusion in Section~\ref{sec:conclusion}. To maintain conciseness and emphasise the results, formal proofs are provided at the end in Sections~\ref{sec:proof-posit-integer}, \ref{sec:proof-posit-consecutive_integers}, \ref{sec:proof-takum-integer}, and \ref{sec:proof-linear_takum-consecutive_integers}.
\section{Posit Encoding Scheme}\label{sec:posit_encoding}
The posit number format, introduced by Gustafson et al. in \cite{posits-beating_floating-point-2017}, has since been extensively studied as a potential replacement for the IEEE 754 standard, both for its numerical properties \cite{posit-dnn-2019, posits-good-bad-ugly-2019} and its applicability to hardware implementations \cite{2024-log-posit, percival-64_bit-2024}. These investigations have culminated in the development of an initial standardisation effort \cite{posits-standard-2022}. The key design feature of posits is their variable-length exponent coding, which allocates additional fraction bits to values close to 1 compared to IEEE 754 floating-point numbers. This improvement is achieved at the expense of reducing the fraction bits for numbers further from 1, particularly those approaching zero or infinity. The format is formally defined as follows:
\begin{definition}[posit encoding\cite{posits-beating_floating-point-2017, posits-standard-2022}]\label{def:posit}
Let $n \in \mathbb{N}$ with $n \ge 5$. Any
$n$-bit MSB$\rightarrow$LSB string $P := (\textcolor{p-sign}{S},\textcolor{p-regime}{R},\textcolor{p-regime-term}{\overline{R_0}},
\textcolor{p-exponent}{E},\textcolor{p-fraction}{F}) \in {\{0,1\}}^n$ of the form
\begin{center}
	\begin{tikzpicture}
		\draw[<->] (0.0, 0.7) -- (0.4, 0.7) node[above,pos=.5] {sign};
		\draw[<->] (0.4, 0.7) -- (4.5, 0.7) node[above,pos=.5] {exponent};
		\draw[<->] (4.5, 0.7) -- (8.0, 0.7) node[above,pos=.5] {fraction};

		\draw (0,  0  ) rectangle (0.4,0.5) node[pos=.5] {\textcolor{p-sign}{$S$}};
		\draw (0.4,0  ) rectangle (3.2,0.5) node[pos=.5] {\textcolor{p-regime}{$R$}};
		\draw (3.2,0  ) rectangle (3.7,0.5) node[pos=.5] {\textcolor{p-regime-term}{$\overline{R_0}$}};
		\draw (3.7,0  ) rectangle (4.5,0.5) node[pos=.5] {\textcolor{p-exponent}{$E$}};
		\draw (4.5,0  ) rectangle (8.0,0.5) node[pos=.5] {$F$};

		\draw[<->] (0.0, -0.2) -- (0.4, -0.2) node[below,pos=.5] {$1$};
		\draw[<->] (0.4, -0.2) -- (3.2, -0.2) node[below,pos=.5] {$k$};
		\draw[<->] (3.2, -0.2) -- (3.7, -0.2) node[below,pos=.5] {$1$};
		\draw[<->] (3.7, -0.2) -- (4.5, -0.2) node[below,pos=.5] {$2$};
		\draw[<->] (4.5, -0.2) -- (8.0, -0.2) node[below,pos=.5] {$p$};
	\end{tikzpicture}
\end{center}
with sign bit $\textcolor{p-sign}{S}$, regime bits
$\textcolor{p-regime}{R} := (\textcolor{p-regime}{R}_{k-1},\dots,\textcolor{p-regime}{R}_0)$,
regime termination bit $\textcolor{p-regime-term}{\overline{R_0}}$, regime
\begin{equation}
	r := \begin{cases}
		-k & \textcolor{p-regime}{R}_0 = 0\\
		k-1 & \textcolor{p-regime}{R}_0 = 1,
	\end{cases}
\end{equation}
exponent bits $\textcolor{p-exponent}{E} :=(\textcolor{p-exponent}{E}_{1},
\textcolor{p-exponent}{E}_0)$, exponent $\hat{e} := 2 \textcolor{p-exponent}{E}_1 +
\textcolor{p-exponent}{E}_0$, fraction bit count
$p := n - k - 4 \in \{ 0,\dots,n-5 \}$, fraction
$f := 2^{-p} \sum_{i=0}^{p-1} \textcolor{p-fraction}{F\!}_i 2^i \in [0,1)$ and
\enquote{actual} exponent
\begin{equation}
	e := {(-1)}^{\textcolor{p-sign}{S}} (4 r + \hat{e} + \textcolor{p-sign}{S})
\end{equation}
encodes the posit value
\begin{equation}\label{eq:posit-value}
	\pi(P) \!:=\! \begin{cases}
		\begin{cases}
			0 & \textcolor{p-sign}{S} = 0\\
			\mathrm{NaR} & \textcolor{p-sign}{S} = 1
		\end{cases}
			& \textcolor{p-regime}{R} = \textcolor{p-regime-term}{\overline{R_0}} =
				\textcolor{p-exponent}{E} = \textcolor{p-fraction}{F} = \bm{0} \\
		[(1-3 \textcolor{p-sign}{S}) + f] \cdot
		2^e & \text{otherwise}.
	\end{cases}
\end{equation}
with $\pi \colon {\{0,1\}}^n \mapsto \{ 0,\mathrm{NaR} \} \cup
\pm\left[2^{-4n+8},2^{4n-8}\right]$.
The symbol $\mathrm{NaR}$ (\enquote{not a real})
represents infinity and other non-representable forms. 
Without loss of generality, any bit string shorter than 5 bits is also
included in the definition by assuming the missing bits to be
zero bits (\enquote{ghost bits}). The colour scheme for the different bit
string segments was adopted from the standard \cite{posits-standard-2022}.
\end{definition}
\section{Takum Encoding Scheme}\label{sec:takum_encoding}
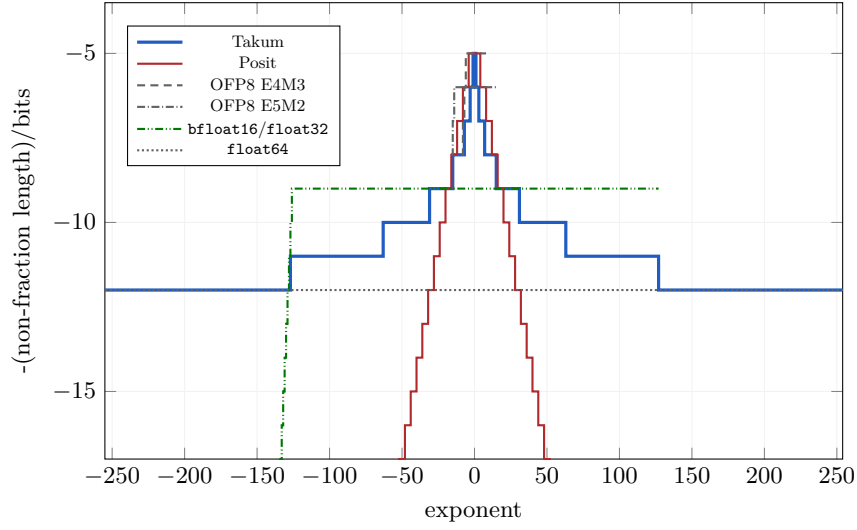
\begin{figure}[tb]
	\begin{center}
		\begin{tikzpicture}
			\begin{axis}[
				scale only axis,
				width=0.8\textwidth,
				height=0.4944\textwidth,
				xlabel={exponent},
				ylabel={-(non-fraction length)/bits},
				ymin=18,
				ymax=31.5,
				xmin=-255,
				xmax=254,
				grid=major,
				ytick={15,20,25,30},
				yticklabels={$-20$,$-15$,$-10$,$-5$},
				grid style={line width=.1pt, draw=gray!10},
				legend style={nodes={scale=0.7, transform shape}},
				legend style={at={(0.03,0.95)},anchor=north west},
			]
				\addplot[const plot,characteristic,very thick] 
				table [x=c, y=takum, col 
				sep=comma]{code/precision.csv};
				\addlegendentry{Takum};
				\addplot[const plot,sign,thick] table [x=c, 
				y=posit, col sep=comma]{code/precision.csv};
				\addlegendentry{Posit};

				\addplot[const plot,mantissa,densely 
				dashed,thick] table [x=c, y=e4m3, col 
				sep=comma]{code/precision.csv};
				\addlegendentry{OFP8 E4M3};

				\addplot[const plot,mantissa,densely 
				dashdotted,thick] table [x=c, y=e5m2, col 
				sep=comma]{code/precision.csv};
				\addlegendentry{OFP8 E5M2};

				\addplot[const plot,direction,densely 
				dashdotdotted,thick] table [x=c, y=float32, col 
				sep=comma]{code/precision.csv};
				\addlegendentry{\texttt{bfloat16}/\texttt{float32}};
				\addplot[const plot,mantissa,densely 
				dotted,thick] table [x=c, y=float64, col 
				sep=comma]{code/precision.csv};
				\addlegendentry{\texttt{float64}};
			\end{axis}
		\end{tikzpicture}
	\end{center}
	\caption{
		Non-fraction length (the number of bits used to encode the sign and exponent, including those fraction bits reassigned in the case of subnormal numbers) as a function of the coded exponent for IEEE 754, OFP8, \texttt{bfloat16}, takum, and posit floating-point formats. The plot is inverted, so larger values correspond to more available fraction bits.
	}
	\label{fig:precision}
\end{figure}
A primary criticism of the posit format is its limited dynamic range and the sharp decline in precision for numbers farther from 1, which arises from the rapidly increasing length of the encoded exponent \cite{posits-good-bad-ugly-2019,2024-takum}. In response, the takum number format has been recently proposed, featuring an alternative exponent coding scheme. This approach trades some density near 1, compared to posits, to achieve greater precision for values farther from 1 (see Figure~\ref{fig:precision}). Specifically, this paper focuses on linear takums, a variant that adopts a floating-point representation, as opposed to takums, a logarithmic number system. The linear takum encoding is formally defined as follows:
\begin{definition}[linear takum encoding {\cite[Definition~8]{2024-takum}}]\label{def:linear_takum}
	Let $n \in \mathbb{N}$ with $n \ge 12$. Any $n$-bit MSB$\rightarrow$LSB 
	string 
	$T := (\textcolor{sign}{S},\textcolor{direction}{D},\textcolor{regime}{R},
	\textcolor{characteristic}{C},\textcolor{mantissa}{F}) \in {\{0,1\}}^n$ of 
	the form
	
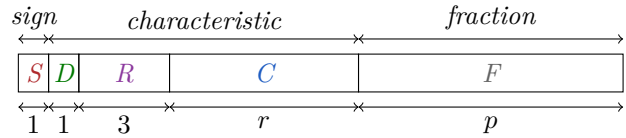
\begin{center}
		\begin{tikzpicture}
			\draw[<->] (0.0, 0.7) -- (0.4, 0.7) node[above,pos=.5] {sign};
			\draw[<->] (0.4, 0.7) -- (4.5, 0.7) node[above,pos=.5] 
			{characteristic};
			\draw[<->] (4.5, 0.7) -- (8.0, 0.7) node[above,pos=.5] {fraction};
			
			\draw (0,  0  ) rectangle (0.4,0.5) node[pos=.5] 
			{\textcolor{sign}{S}};
			\draw (0.4,0  ) rectangle (0.8,0.5) node[pos=.5] 
			{\textcolor{direction}{D}};
			\draw (0.8,0  ) rectangle (2.0,0.5) node[pos=.5] 
			{\textcolor{regime}{R}};
			\draw (2.0,0  ) rectangle (4.5,0.5) node[pos=.5] 
			{\textcolor{characteristic}{C}};
			\draw (4.5,0  ) rectangle (8.0,0.5) node[pos=.5] 
			{\textcolor{mantissa}{F}};
			
			\draw[<->] (0.0, -0.2) -- (0.4, -0.2) node[below,pos=.5] {$1$};
			\draw[<->] (0.4, -0.2) -- (0.8, -0.2) node[below,pos=.5] {$1$};
			\draw[<->] (0.8, -0.2) -- (2.0, -0.2) node[below,pos=.5] {$3$};
			\draw[<->] (2.0, -0.2) -- (4.5, -0.2) node[below,pos=.5] {$r$};
			\draw[<->] (4.5, -0.2) -- (8.0, -0.2) node[below,pos=.5] {$p$};
		\end{tikzpicture}
	\end{center}
	with {sign bit} $\textcolor{sign}{S}$, {direction bit}
	$\textcolor{direction}{D}$, {regime bits}
	$\textcolor{regime}{R} := (\textcolor{regime}{R}_2,
	\textcolor{regime}{R}_1,\textcolor{regime}{R}_0)$, characteristic 
	bits $\textcolor{characteristic}{C} :=(\textcolor{characteristic}{C}_{r-1},\dots,
	\textcolor{characteristic}{C}_0)$, {fraction bits}
	$\textcolor{mantissa}{F} := (\textcolor{mantissa}{F}_{p-1},\dots,
	\textcolor{mantissa}{F}_0)$, {regime}
	\begin{equation}
		r := \begin{cases}
			\integer(\overline{\textcolor{regime}{R}}) & 
			\textcolor{direction}{D} = 0\\
			\integer(\textcolor{regime}{R})	&
			\textcolor{direction}{D} = 1
		\end{cases}
		\in \{0,\dots,7\},
	\end{equation}
	{characteristic}
	\begin{equation}\label{eq:linear_takum-characteristic}
		c :=
		\begin{cases}
			-2^{r+1} + 1 + \integer(\textcolor{characteristic}{C}) & 
			\textcolor{direction}{D} = 
			0\\
			2^r - 1 + \integer(\textcolor{characteristic}{C})
			& \textcolor{direction}{D} = 1
		\end{cases}
		\in \{ -255,\dots,254 \},
	\end{equation}
	{fraction bit count} $p := n - r -5 \in \{n-12,\dots,n-5\}$,
	{fraction} $f := 2^{-p} \integer(\textcolor{mantissa}{F})
	\in [0,1)$ and {exponent}
	\begin{equation}\label{eq:linear_takum-exponent}
		e := {(-1)}^{\textcolor{sign}{S}}(c + \textcolor{sign}{S}) \in \{-255,254\}
	\end{equation}
	encodes the linear takum value
	\begin{equation}\label{eq:linear_takum-value}
		\lineartakum(T)
		:= \begin{cases}
			\begin{cases}
				0 & \textcolor{sign}{S} = 0\\
				\mathrm{NaR} & \textcolor{sign}{S} = 1
			\end{cases}
			& \textcolor{direction}{D} = \textcolor{regime}{R} = 
			\textcolor{characteristic}{C} = \textcolor{mantissa}{F} 
			= \bm{0} \\
			\left[ (1-3\textcolor{sign}{S}) + f \right] \cdot 2^e & \text{otherwise}
		\end{cases}
	\end{equation}
	with $\lineartakum \colon {\{0,1\}}^n \mapsto \{ 0,\mathrm{NaR} \} \cup
	\pm\left(2^{-255},2^{255}\right)$. The symbol $\mathrm{NaR}$ (\enquote{not a real})
	represents infinity and other non-representable forms. 
	Without loss of generality, any bit string shorter than 12 bits is also
	considered in the definition by assuming the missing bits to be
	zero bits (\enquote{ghost bits}). The colour scheme for the different bit
	string segments was adopted from the specification \cite{2024-takum}.
\end{definition}
Takums and posits are defined for any $n \in \mathbb{N}_1$, offering a degree of flexibility not present in IEEE 754 floating-point numbers, where $n$ is restricted to $\{16, 32, 64, \dots\}$.
This also stands in contrast to OFP8 formats,
which are fixed at $n=8$.
As follows we will refer to linear takums as takums.
\section{Integer Representations}\label{sec:integer_representations}
The analysis of integer representations in floating-point formats is governed by a straightforward principle: the integer bits are encoded within the fraction, and the exponent is adjusted to shift the decimal point sufficiently far so that all fraction bits lie to the left of it.
\par
A key challenge arises from the role of implicit zero bits to the right of the explicit fraction bits, which may form part of the integer's binary representation. Another, more nuanced challenge pertains to the analysis of tapered-precision floating-point formats. Unlike fixed-precision formats such as IEEE 754, tapered-precision formats feature a variable number of fraction bits depending on the exponent value (see Figure~\ref{fig:precision}). This variability introduces an intricate balance between exponent value and number of fraction bits that is absent in traditional fixed-precision floating-point formats.
\subsection{IEEE 754 and Bfloat16}
We begin by examining the IEEE 754 standard and the compatible \texttt{bfloat16} format. Although the following result may appear elementary, it does not seem to have been explicitly documented in the existing literature. To ensure completeness, we present the result along with a proof, providing a comprehensive foundation for this paper. However, for the sake of simplicity, we refrain from providing a fully formalised proof, as this would necessitate a more extensive introduction to the IEEE 754 floating-point representation.
\begin{proposition}[Consecutive IEEE 754 Integers]\label{prop:ieee_754-consecutive_integers}
	Let an IEEE 754 floating-point format with $n_e$ (explicit)
	exponent and $n_f$ fraction bits, and $n_e$ sufficiently large for
	the exponent to assume the value $n_f + 1$. All
	$m \in \mathbb{Z}$ with
	\begin{equation}\label{eq:ieee_754-consecutive_bound}
		|m| \le 2^{n_f + 1}
	\end{equation}
	are exactly representable in the format.
\end{proposition}
\begin{proof}
	Consider a number in the specified format where all the fraction bits are set to one and the exponent value is $n_f$. In this configuration, the number represents the largest explicit integer that can be formed within this format, which is $2^{n_f+1} - 1$. By observing the structure of the representation, it follows that all preceding positive integers up to $2^{n_f+1} - 1$ can also be exactly represented.
	\par
	Next, consider $2^{n_f+1}$. This value is exactly representable with an exponent value of $n_f+1$ (as required to be achievable) and all fraction bits set to zero. However, note that the least significant zero bit of $2^{n_f+1}$ is implicit due to the encoding scheme. This implicitness has a key implication: the number $2^{n_f+1} + 1$, which would require toggling that implicit bit to one, cannot be exactly represented within the given format.
	Thus, $2^{n_f+1} + 1$ is not exactly representable.
	\qed
\end{proof}
It shall be noted that the condition on $n_e$ being sufficiently
large is met by all standard formats and \texttt{bfloat16} but included for completeness.
This result will serve as the reference point for the subsequent analysis of posits and takums.
\subsection{OFP8}
Despite certain deviations from the IEEE 754 standard, the results derived in Proposition~\ref{prop:ieee_754-consecutive_integers} remain applicable, since both IEEE 754 and OFP8 define normal numbers in the same way, and all integral values representable in these floating-point formats are normal. We therefore apply the proposition to both OFP8 types.  
\par  
For the OFP8 E4M3 type, which has four exponent bits and three fraction bits, the biased exponent ranges from $-6$ to $8$ \cite{ofp8}. Hence, $n_e = 3$ and $n_f = 4$. Since $n_f + 1 = 5$ lies within the exponent range, Proposition~\ref{prop:ieee_754-consecutive_integers} implies that all $m \in \mathbb{Z}$ with $|m| \leq 2^{n_f+1} = 2^5 = 32$ are exactly representable in E4M3.  
Similarly, the OFP8 E5M2 type has five exponent bits and two fraction bits, giving a biased exponent range from $-14$ to $15$ \cite{ofp8}. Thus, $n_e = 5$ and $n_f = 2$. Since $n_f + 1 = 3$ also lies within the exponent range, it follows that all $m \in \mathbb{Z}$ with $|m| \leq 2^{n_f+1} = 2^3 = 8$ are exactly representable in E5M2.
\subsection{Posit Arithmetic}
In the case of tapered-precision formats such as posits, and subsequently takums, we adopt a proof strategy that proceeds in two stages. First, we establish how many bits are required to precisely represent a given integer, a result that also holds intrinsic value and utility. In the second stage, we leverage this result to determine the largest consecutive integer that can be represented for a fixed $n$. We begin with the following result:
\begin{proposition}[Posit Integer Representation]\label{prop:posit-integer}
	Let $m \in \mathbb{Z} \setminus \{0\}$ with
	$v := 1 + \lfloor\log_2(|m|)\rfloor$ bits and
	$w := \max_{i \in \mathbb{N}_0} \left(2^i \mid m\right)$
	trailing zeros in $|m|$'s binary representation.
	There exists an $M \in {\{0,1\}}^\ell$ with $\pi(M)=m$ and
	\begin{equation}\label{eq:prop-posit-integers-ell}
		\ell := {\left\lfloor \frac{5 (v+3)}{4} - w \right\rfloor} -
				(w = v-1) \cdot \begin{cases}
					3 & v \in 4\mathbb{N}_0+1\\
					1 & v \in 4\mathbb{N}_0+3,
				\end{cases}
	\end{equation}
	which is the shortest possible representation.
\end{proposition}
\begin{proof}
	See Section~\ref{sec:proof-posit-integer}.
	\qed
\end{proof}
As can be observed, $\ell$ is not merely a bound on the length but represents the exact length. Although this results in a complex expression due to the special cases involving zero-bit truncation, this complexity is essential for establishing the proof of the following proposition:
\begin{proposition}[Consecutive Posit Integers]\label{prop:posit-consecutive_integers}
	Let $n \in \mathbb{N}_3$. It holds for all $m \in \mathbb{Z}$ with
	\begin{equation}\label{eq:posit-consecutive_bound}
		|m| \le 2^{\left\lfloor \frac{4(n-3)}{5} \right\rfloor}
	\end{equation}
	that there exists an $M \in {\{0,1\}}^n$ with $m = \pi(M)$.
\end{proposition}
\begin{proof}
	See Section~\ref{sec:proof-posit-consecutive_integers}.
	\qed
\end{proof}
Although the complete evaluation is deferred to Section~\ref{sec:evaluation}, it can already be noted that both IEEE 754 floating-point numbers and posits share the property that the largest consecutive posit integer is a power of two.
\subsection{Takum Arithmetic}
Following an approach analogous to that used for posits, we first establish the number of bits required to represent a given integer as a takum.
\begin{proposition}[Takum Integer Representation]\label{prop:takum-integer}
	Let $m \in \mathbb{Z} \setminus \{0\}$ with
	$|m| \le 2^{254}$,
	$v := 1 + \lfloor\log_2(|m|)\rfloor$ bits and
	$w := \max_{i \in \mathbb{N}_0} \left(2^i \mid m\right)$
	trailing zeros in $|m|$'s binary representation.
	There exists an $M \in {\{0,1\}}^\ell$ with $\lineartakum(M)=m$ and
	\begin{align}\label{eq:prop-takum-integers-ell}
		\ell &:= {\left\lfloor 4+v+\log_2(v) - w \right\rfloor} -\notag\\
			&\hphantom{\Leftrightarrow}\quad (w = v-1) \cdot \max_{i \in \{ 0,\dots, \lfloor \log_2(v) \rfloor\}}
				\left( 2^i \mid v - 2^{\lfloor \log_2(v) \rfloor} \right) -\notag\\
			&\hphantom{\Leftrightarrow}\quad (v \in 2^{\mathbb{N}_0}) \cdot \max_{i \in \{0,\dots,3\}}
				\left( 2^i \mid \lfloor \log_2(v) \rfloor \right),
	\end{align}
	which is the shortest possible representation.
\end{proposition}
\begin{proof}
	See Section~\ref{sec:proof-takum-integer}.
	\qed
\end{proof}
Compared to the posit result in Proposition~\ref{prop:posit-integer}, $m$ is limited in range and the expression for $\ell$ accounts for two special cases instead of one. This distinction arises because, in the case of posits, zero-bit truncation can occur only in the exponent bits. In contrast, for takums, truncation may occur in both the characteristic and regime bits, adding a layer of complexity to the analysis. Building on this result, we derive the following proposition.
\begin{proposition}[Consecutive Takum Integers]\label{prop:linear_takum-consecutive_integers}
	Let $n \in \mathbb{N}_5$. It holds for all $m \in \mathbb{Z}$ with $|m| \le 2^{254}$ and
	\begin{equation}\label{eq:takum-consecutive_bound}
		|m| \le
		2^{\left\lceil \frac{W_0\!\left(2^{n-3} \ln(2)\right)}{\ln(2)} - 1 \right\rceil},
	\end{equation}
	where $W_0$ is the principal branch of the \textsc{Lambert} $W$ function,
	that there exists an $M \in {\{0,1\}}^n$ with $m = \lineartakum(M)$.
\end{proposition}
\begin{proof}
	See Section~\ref{sec:proof-linear_takum-consecutive_integers}.
	\qed
\end{proof}
Despite the greater complexity of (\ref{eq:takum-consecutive_bound}) compared to the corresponding expressions for IEEE 754 in (\ref{eq:ieee_754-consecutive_bound}) and for posits in (\ref{eq:posit-consecutive_bound}), most notably due to the inclusion of a non-analytical function, the upper bound remains a power of two, consistent with the other two formats.
\section{Evaluation}\label{sec:evaluation}
\begin{figure}[tb]
	\begin{center}
		\begin{tikzpicture}
			\begin{axis}[
				scale only axis,
				xmode=log,
				width=0.8\textwidth,
				height=0.4944\textwidth,
				xtick={5,8,16,32,64,128},
				xticklabels={$5$, $8$, $16$, $32$, $64$, $128$},
				ytick={0,16.61,33.22,49.83,66.44,83.05,99.66,116.27},
				yticklabels={$0$,$5$,$10$,$15$,$20$,$25$,$30$,$35$},
				minor ytick={16.61,49.83,83.05,116.27},
				xlabel={$n$},
				ylabel={$\log_{10}(\text{largest consecutive integer})$},
				xlabel near ticks,
				grid=both,
				legend style={nodes={scale=0.7, transform shape}},
				legend style={at={(0.03,0.95)},anchor=north west},
				grid style={line width=.1pt, draw=gray!10},
			]
				\addplot [mantissa,only marks,mark 
				size=1.3pt,thick] table [col sep=comma] {
					16,11
					32,24
					64,53
					128,113
				};
				\addlegendentry{IEEE 754}

				\addplot [mantissa,only marks,mark=x,mark 
				size=3pt,thick] coordinates {(8,5)};
				\addlegendentry{OFP8 E4M3}

				\addplot [mantissa,only marks,mark=Mercedes 
				star,mark size=3pt,thick] coordinates {(8,3)};
				\addlegendentry{OFP8 E5M2}

				\addplot [direction,only marks,mark=star,mark 
				size=3pt,thick] coordinates {(16,8)};
				\addlegendentry{\texttt{bfloat16}}

				\addplot [characteristic,very thick] table [col 
				sep=comma,x=n,y=takum_exponent] {code/data.csv};
				\addlegendentry{Takum}
				\addplot [sign,thick] table [col sep=comma, 
				x=n, y=posit_exponent] {code/data.csv};
				\addlegendentry{Posit}
			\end{axis}
		\end{tikzpicture}
	\end{center}
	\caption{
		The largest consecutive integers for IEEE 754 for $n \in \{ 16,32,64,128\}$, OFP8, \texttt{bfloat16}, takums and posits
		relative to the bit string length $n$.
	}
	\label{fig:plot}
\end{figure}
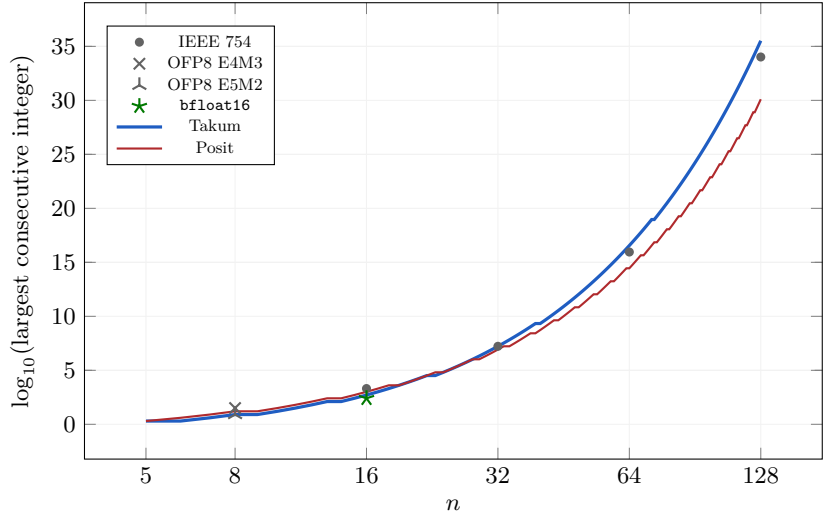
\begin{table}[tb]
	\begin{center}
		\bgroup
		\def\arraystretch{1.3}
		\setlength{\tabcolsep}{0.3em}
		\begin{tabular}{| l | l | l |}
			\hline
			\textbf{type} & \textbf{largest consecutive integer} & \textbf{signed integer ratio/\%}\\\hline\hline
			OFP8 E4M3                & $2^{5\phantom{1}} = 32$ & $\num{25}$\\\hline
			OFP8 E5M2                & $2^{3\phantom{1}} = 8$ & $\num{6.3}$\\\hline
			\texttt{posit8}         & $2^{4\phantom{1}} = 16$ & $\num{13}$\\\hline
			\texttt{takum8} & $2^{3\phantom{1}} = 8$ & $\num{6.3}$\\\hline\hline
			\texttt{float16}         & $2^{11} = 2048$ & $\num{6.3}$\\\hline
			\texttt{bfloat16}        & $2^{8\phantom{1}} = 256$ & $\num{0.78}$\\\hline
			\texttt{posit16}         & $2^{10} = 1024$ & $\num{3.1}$\\\hline
			\texttt{takum16} & $2^{9\phantom{1}} = 512$ & $\num{1.6}$\\\hline\hline
			\texttt{float32}         & $2^{24} \approx \num{1.7e7}$ & $\num{0.78}$\\\hline
			\texttt{posit32}         & $2^{23} \approx \num{8.4e6}$ & $\num{0.39}$\\\hline
			\texttt{takum32} & $2^{24} \approx \num{1.7e7}$ & $\num{0.78}$\\\hline\hline
			\texttt{float64}         & $2^{53} \approx \num{9.0e15}$ & $\num{9.8e-2}$\\\hline
			\texttt{posit64}         & $2^{48} \approx \num{2.8e14}$ & $\num{3.1e-3}$\\\hline
			\texttt{takum64} & $2^{55} \approx \num{3.6e16}$ & $\num{0.39}$\\\hline\hline
			\texttt{float128}         & $2^{113} \approx \num{1.0e34}$ & $\num{6.1e-3}$\\\hline
			\texttt{posit128}         & $2^{100} \approx \num{1.3e30}$ & $\num{7.5e-7}$\\\hline
			\texttt{takum128} & $2^{118} \approx \num{3.3e35}$ & $\num{0.20}$\\\hline
		\end{tabular}
		\egroup
	\end{center}
	\caption{Largest consecutive integers for IEEE 754 for $n \in \{ 16,32,64,128\}$, OFP8, \texttt{bfloat16}, takum and posit
	formats. The signed integer ratio is the ratio of the largest consecutive integer in each floating-point format and the largest two's complement signed integer $2^{n-1}-1$ of the same bit length $n$.}
	\label{tab:table}
\end{table}
Although the proofs involve considerable complexity, the evaluation ultimately reduces to a comparison of the results, specifically, the values of the largest consecutive integers, presented in Propositions~\ref{prop:ieee_754-consecutive_integers}, \ref{prop:posit-consecutive_integers}, and \ref{prop:linear_takum-consecutive_integers}. These results are visualised for $n \in \{5,\dots,128\}$ in Figure~\ref{fig:plot} and tabulated for $n \in \{8,16,32,64,128\}$ in Table~\ref{tab:table}.
\par
At eight bits, the OFP8 E4M3 type performs best, closely followed by \texttt{posit8}, with \texttt{takum8} and E5M2 on par. At $16$ bits, \texttt{float16} achieves the best performance and \texttt{bfloat16} the worst, while \texttt{posit16} and \texttt{takum16} lie in between, the former slightly outperforming the latter. At $32$ bits, \texttt{posit32} lags behind both \texttt{takum32} and \texttt{float32}, which perform equivalently. Beyond $32$ bits, a clear trend emerges: takums significantly outperform both IEEE 754 floats and posits, with the latter falling considerably behind. For instance, at $64$ bits, \texttt{takum64} outperforms IEEE 754 by approximately half an order of magnitude and \texttt{posit64} by two orders of magnitude. This trend continues at $128$ bits, with performance gains of one and five orders of magnitude, respectively.  
\par  
The signed integer ratio reported in Table~\ref{tab:table} for each type demonstrates that the use of floating-point number formats for representing integral data is, overall, highly inefficient. Nonetheless, notable differences exist among the formats. While the OFP8 E4M3 type achieves \SI{25}{\percent} of the efficiency of an 8-bit signed integer, its counterpart E5M2, together with \texttt{takum8}, attains only \SI{6.3}{\percent}. At $16$ and $32$ bits, all formats remain within one order of magnitude of each other, indicating relatively balanced performance. From $64$ bits onwards, however, distinct differences emerge: IEEE 754 types and, to an even greater extent, posits exhibit a marked decline relative to the corresponding signed integer, whereas takums largely preserve their relative performance with respect to the corresponding signed integer.
\section{Conclusion}\label{sec:conclusion}
In this paper, we formally analysed the capacity of IEEE 754, OFP8, \texttt{bfloat16}, posit, and takum arithmetic to represent integral values, with particular emphasis on consecutive integer representations.
Given that \texttt{bfloat16} is a more practical and widely adopted reference for low-precision formats than \texttt{float16}, and that the two OFP8 types represent opposite extremes in performance, our results show that although posits slightly outperform takums at $8$ and $16$ bits, takums reach parity with posits at $32$ bits and increasingly surpass them at higher precisions, thereby demonstrating overall superiority across all practically relevant formats.
\par  
Takums match or exceed the performance of both \texttt{bfloat16}, the canonical reference, and IEEE 754 floating-point numbers at $32$, $64$, and $128$ bits, whereas posits only marginally outperform \texttt{bfloat16} before falling significantly behind all IEEE 754 floating-point types. This demonstrates that takums are fully backwards-compatible with IEEE 754, in contrast to posits. When compared with signed integers through the signed integer ratio metric, all formats except takums exhibit a sharp and increasing decline in performance beyond $32$ bits, whereas takums maintain a balanced performance.
\par  
The formal methodology adopted in this study not only corroborates empirical findings for posits reported in \cite{posits-standard-2022}, but also extends the analysis to takum arithmetic, for which the complexity of the consecutive integer bounds would render a purely empirical approach infeasible. Furthermore, we provide a formal proof for IEEE 754 floating-point numbers which, to the best of the authors’ knowledge, has not previously appeared in the literature.
\par  
It might be argued that it constitutes a disadvantage that the proofs for posits and takums are considerably more complex than those for IEEE 754 floats. However, in the author's view, what ultimately matters is the numerical and practical performance of a number format. By analogy, few would argue for the superiority of the \textsc{Euler} method over the \textsc{Runge}--\textsc{Kutta}--\textsc{Fehlberg} method.  
\par 
In conclusion, the OFP8 types have been shown to yield mixed results, \texttt{bfloat16} generally performs poorly, and takum arithmetic demonstrates balanced performance while consistently matching or surpassing IEEE 754 floating-point numbers in terms of consecutive integer representation. This marks a clear distinction from posit arithmetic, which falls short in this regard. The finding is particularly significant given the extensive body of work founded on IEEE 754, much of which implicitly relies on its integral representation capabilities. By preserving these established properties while offering greater flexibility, takum arithmetic not only positions itself as a compelling alternative, but also as a promising successor to IEEE 754 for future numerical frameworks.
\section{Proof of Proposition~\ref{prop:posit-integer}}\label{sec:proof-posit-integer}
Without loss of generality we can assume $m \in \mathbb{N}_1$, because the
set of $n$-bit posits is closed under negation. We know that $m$ has the form
\begin{equation}\label{eq:proof-posit_cost-m_representation}
	m = 2^{v-1} + \sum_{i=w}^{v-2} {\textcolor{p-fraction}{F}\!}_{i-w} 2^i
\end{equation}
with $\textcolor{p-fraction}{F} \in {\{ 0,1 \}}^{v-w-1}$ and
${\textcolor{p-fraction}{F}\!}_0 \neq 0$. Its corresponding floating-point
representation follows directly as
\begin{equation}\label{eq:proof-posit_cost-floating-point_representation}
	m = (1 + f) \cdot 2^{v-1}
\end{equation}
with $f := \sum_{i=w}^{v-2} {\textcolor{p-fraction}{F}\!}_{i-w} 2^{i-v+1} \in [0,1)$.
For (\ref{eq:proof-posit_cost-floating-point_representation}) to correspond
with (\ref{eq:posit-value}) it must hold
\begin{equation}
	v - 1 = \hat{e} = 4 (k-1) + e \Leftrightarrow
	k = \frac{v+3-e}{4}
\end{equation}
with $k \in \mathbb{N}_1$ and $e \in \{0,\dots,3\}$. We set
$e = (v+3) \bmod 4$ and obtain $k = \left\lfloor \frac{v+3}{4} \right\rfloor$.
The posit bit representation of $m$ follows as
$M := (\textcolor{p-sign}{0}, \textcolor{p-regime}{\bm{1}_k}, \textcolor{p-regime-term}{0}, \textcolor{p-exponent}{E}, \textcolor{p-fraction}{F})$
with non-reduced length
\begin{equation}
	1 + \left\lfloor \frac{v+3}{4} \right\rfloor + 1 + 2 + (v - w - 1) =
	\left\lfloor \frac{5 (v+3)}{4} - w \right\rfloor.
\end{equation}
We know that $M$ is the shortest possible representation when
$\textcolor{p-fraction}{F}$ has non-zero length, because ${\textcolor{p-fraction}{F}\!}_0 \neq 0$.
Otherwise it holds $v-w-1 = 0 \Leftrightarrow w = v -1$ and with
(\ref{eq:proof-posit_cost-m_representation}) it follows
$m = 2^{v-1}$. We now check each possible value of
$\textcolor{p-exponent}{E}$ (corresponding to $e$) to assess
the reducibility of $M$.
\begin{description}
	\item[Case 1 ($e=0$)]{
		~\\This implies $\textcolor{p-exponent}{E}=\bm{0}_2$, and
		given $\textcolor{p-regime-term}{\overline{R_0}} = 0$ it
		follows that $M$ can be reduced by 3 bits. With
		$e = (v+3) \bmod 4$ this case is equivalent to
		$v \in 4 \mathbb{N}_0 + 1$.
	}
	\item[Case 2 ($e=1$)]{
		~\\This implies $\textcolor{p-exponent}{E}=(0,1)$, which
		means that $M$ cannot be further reduced.
	}
	\item[Case 3 ($e=2$)]{
		~\\This implies $\textcolor{p-exponent}{E}=(1,0)$, which
		means that $M$ can be reduced by 1 bit. With
		$e = (v+3) \bmod 4$ this case is equivalent to
		$v \in 4 \mathbb{N}_0 + 3$.
	}
	\item[Case 4 ($e=3$)]{
		~\\This implies $\textcolor{p-exponent}{E}=(1,1)$, which
		means that $M$ cannot be further reduced.
	}
\end{description}
No further reduction is possible, as the direction bit $\textcolor{direction}{D}$
is always $1$. The reduced length of $M$ follows as
\begin{equation}
	{\left\lfloor \frac{5 (v+3)}{4} - w \right\rfloor} -
	(w= v-1) \cdot \begin{cases}
		3 & v \in 4\mathbb{N}_0+1\\
		1 & v \in 4\mathbb{N}_0+3,
	\end{cases}
\end{equation}
which was to be shown.\qed
\section{Proof of Proposition~\ref{prop:posit-consecutive_integers}}\label{sec:proof-posit-consecutive_integers}
Without loss of generality we can assume $m \in \mathbb{N}_1$, because the set of
$n$-bit posits is closed under negation and the integer zero is represented by
$\bm{0}_n$ for any $n$. Let us further assume that $m$ is an arbitrary $v$-bit integer
with $w$ trailing zeros in its binary representation, where $v \in \mathbb{N}_1$ and
$w \in \mathbb{N}_0$.
Our goal is to determine an upper bound for $v$ that depends on $n$ to find out the
largest consecutive integer represented by an $n$-bit posit.
\par
With Proposition~\ref{prop:posit-integer} we know there exists an $M \in {\{0,1\}}^\ell$ with $\pi(M)=m$ and $\ell$ as in (\ref{eq:prop-posit-integers-ell}). It holds
with $v \in \mathbb{N}_1$ (operation 1) that
\begin{align}
	\ell \le n
	&\Leftrightarrow {\left\lfloor \frac{5 (v+3)}{4} - w \right\rfloor} -\notag\\
	&\hphantom{\Leftrightarrow}\quad (w= v-1) \cdot
		\begin{cases}
			3 & v \in 4\mathbb{N}+1\\
			1 & v \in 4\mathbb{N}+3
		\end{cases} \le n\label{eq:proof-posit-consecutive_integers-before}\\
	&\Leftarrow {\left\lfloor \frac{5 (v+3)}{4} \right\rfloor} \le n\label{eq:proof-posit-consecutive_integers-after}\\
	&\Leftrightarrow \frac{5 (v+3)}{4} < n+1\\
	&\Leftrightarrow v < \frac{4(n+1)}{5} - 3\\
	&\Leftrightarrow v < \frac{4n-11}{5}\\
	&\overset{1}{\Leftrightarrow} v \le \left\lceil \frac{4n-11}{5} - 1 \right\rceil\\
	&\Leftrightarrow v \le \left\lceil \frac{4n-16}{5} \right\rceil\\
	&\Leftrightarrow v \le \left\lfloor \frac{4n-16 + (5-1)}{5} \right\rfloor\\
	&\Leftrightarrow v \le \left\lfloor \frac{4(n-3)}{5} \right\rfloor.
\end{align}
Thus we have found an upper bound on $v$ such that arbitrary $v$-bit integers
are represented by an $n$-bit posit. Given the step from (\ref{eq:proof-posit-consecutive_integers-before}) to (\ref{eq:proof-posit-consecutive_integers-after}) this upper bound is not
tight and only a starting point for the next stage of the proof.
\par
Let $v= \left\lfloor \frac{4(n-3)}{5} \right\rfloor$ and $m=2^v-1$,
namely a $v$-bit saturated integer and thus the largest consecutive integer
value representable with an $n$-bit posit that we know of.
Now let us take a look at $m+1 = 2^v$, a $v+1$-bit integer with $v$ trailing
zeros in its binary representation. With Proposition~\ref{prop:posit-integer}
we know that there exists $M'\in {\{0,1\}}^{\ell'}$ with $\pi(M')=m+1$
and
{\allowdisplaybreaks
\begin{align}
	\ell' &:= {\left\lfloor \frac{5 ((v+1)+3)}{4} - v \right\rfloor} -\\
	&\hphantom{=}\quad (v = (v+1)-1) \cdot \begin{cases}
			3 & v+1 \in 4\mathbb{N}_0+1\\
			1 & v+1 \in 4\mathbb{N}_0+3,
		\end{cases}\\
	&\le \left\lfloor \frac{v}{4} + 5 \right\rfloor -
		\begin{cases}
			3 & v \in 4\mathbb{N}_0\\
			1 & v \in 4\mathbb{N}_0+2,
		\end{cases}\\
	&= \left\lfloor \frac{n-3}{5} + 5 \right\rfloor -
		\begin{cases}
			3 & v \in 4\mathbb{N}_0\\
			1 & v \in 4\mathbb{N}_0+2,
		\end{cases}\\
	&\begin{cases}
			=2 & n = 4\\
			=5 & n = 5\\
			< \frac{n-3}{5} + 5 & n \ge 6
		\end{cases}\\
	&\le n.
\end{align}
}
As we can see $m+1$ has a posit representation that fits within $n$
bits, which means that it is also part of the set of representable
consecutive integers.
Let us now take a look at $m+2 = 2^v+1$, a $(v+1)$-bit integer with
zero trailing zeros in its binary representation ($w=0$). With
Proposition~\ref{prop:posit-integer} we know that there exists
$M'' \in {\{0,1\}}^{\ell''}$ with $\pi(M'')=m+2$ and, with
$\mathbb{N}_1 \ni v \neq 0$ (operation 1),
\begin{align}
	\ell'' &:= {\left\lfloor \frac{5 ((v+1)+3)}{4} \right\rfloor} -\\
	&\hphantom{=}\quad (0 = (v+1)-1) \cdot \begin{cases}
			3 & v+1 \in 4\mathbb{N}_0+1\\
			1 & v+1 \in 4\mathbb{N}_0+3,
		\end{cases}\\
	&\overset{1}{=} \left\lfloor \frac{5v}{4}+5 \right\rfloor\\
	&= \left\lfloor n-3+5 \right\rfloor\\
	&= n+2\\
	&> n.
\end{align}
Here we can see that m+2's posit representation does not
fit within $n$ bits. Thus $m+1=2^v$ is the largest representable
consecutive integer, as was to be shown.\qed
\section{Proof of Proposition~\ref{prop:takum-integer}}\label{sec:proof-takum-integer}
Without loss of generality we can assume $m \in \mathbb{N}_1$, because the
set of $n$-bit takums is closed under negation. By construction $m$ is in the image
of $\lineartakum$. We know that $m$ has the form
\begin{equation}\label{eq:proof-takum_cost-m_representation}
	m = 2^{v-1} + \sum_{i=w}^{v-2} {\textcolor{mantissa}{F}\!}_{i-w} 2^i
\end{equation}
with $\textcolor{mantissa}{F} \in {\{ 0,1 \}}^{v-w-1}$ and
${\textcolor{mantissa}{F}\!}_0 \neq 0$. Its corresponding floating-point
representation follows directly as
\begin{equation}\label{eq:proof-takum_cost-floating-point_representation}
	m = (1 + f) \cdot 2^{v-1}
\end{equation}
with $f := \sum_{i=w}^{v-2} {\textcolor{mantissa}{F}\!}_{i-w} 2^{i-v+1} \in [0,1)$.
For (\ref{eq:proof-takum_cost-floating-point_representation}) to correspond
with (\ref{eq:linear_takum-value}) it must hold $\textcolor{sign}{S} = 0$ and
\begin{equation}
	v - 1 = e = {(-1)}^{\textcolor{sign}{S}}(c + \textcolor{sign}{S}) = c \Leftrightarrow
	v = c + 1.
\end{equation}
As $v \ge 1$ it follows $c \ge 0$ and thus $\textcolor{direction}{D} = 1$. This
yields with (\ref{eq:linear_takum-characteristic}) that
\begin{equation}\label{eq:proof-takum_cost-c}
	v = c + 1 = 2^r - 1 + \integer(\textcolor{characteristic}{C}) + 1
	= 2^r + \integer(\textcolor{characteristic}{C}),
\end{equation}
and we can deduce
\begin{equation}\label{eq:proof-takum_cost-regime_value}
	r = \lfloor \log_2(v) \rfloor.
\end{equation}
The takum bit representation of $m$ follows as
$M := (\textcolor{sign}{0},\textcolor{direction}{1},\textcolor{regime}{R},
\textcolor{characteristic}{C},\textcolor{mantissa}{F})$ with non-reduced
length
\begin{equation}
	1 + 1 + 3 + \lfloor \log_2(v) \rfloor + (v - w - 1) =
	\left\lfloor 4 + v + \log_2(v) - w \right\rfloor
\end{equation}
We know that $M$ is the shortest possible representation when
$\textcolor{p-fraction}{F}$ has non-zero length, because ${\textcolor{p-fraction}{F}\!}_0 \neq 0$.
Otherwise it holds $v-w-1 = 0 \Leftrightarrow w = v -1$ and with
(\ref{eq:proof-takum_cost-m_representation}) it follows
$m = 2^{v-1}$. In this case we must check if the characteristic
bits $\textcolor{characteristic}{C}$ have trailing zeros that
can be reduced. With (\ref{eq:proof-takum_cost-c}) and
(\ref{eq:proof-takum_cost-regime_value}) we know that
\begin{equation}\label{eq:proof-takum_cost-characteristic}
	\integer(\textcolor{characteristic}{C}) = v - 2^r =
	v - 2^{\lfloor \log_2(v) \rfloor}.
\end{equation}
The number of trailing zeros in the $r = \lfloor \log_2(v) \rfloor $ characteristic bits
is obtained with
\begin{equation}
	\max_{i \in \{ 0,\dots,r \}}
	\left( 2^i \mid v \! - \! 2^{r} \right) =
	\max_{i \in \{ 0,\dots,\lfloor \log_2(v) \rfloor \}}
	\left( 2^i \mid v \!-\! 2^{\lfloor \log_2(v) \rfloor} \right),
\end{equation}
namely the largest integer $i$ such that $2^i$ divides
$\integer(\textcolor{characteristic}{C})$.
In the extreme case that all characteristic bits are zero,
implying $\integer(\textcolor{characteristic}{C}) = 0$,
we might also be able to reduce trailing zeros in the regime. Using
(\ref{eq:proof-takum_cost-characteristic}) this case
is equivalent to
\begin{equation}
	0 =
	v - 2^{\lfloor \log_2(v) \rfloor} \Leftrightarrow
	v = 2^{\lfloor \log_2(v) \rfloor} \Leftrightarrow
	v \in 2^{\mathbb{N}_0}.
\end{equation}
The number of trailing zero bits in the 3 regime bits
is, analogous to the trailing zero bits in the characteristic bits,
the largest integer $i \in \{ 0,\dots,3 \}$ such that $2^i$ divides
$r$, formally
\begin{equation}
	\max_{i \in \{ 0,\dots,3 \}} \left( 2^i \mid r \right) =
	\max_{i \in \{ 0,\dots,3 \}} \left( 2^i \mid \lfloor \log_2(v) \rfloor \right).
\end{equation}
No further reduction is possible, as the direction bit $\textcolor{direction}{D}$
is always $1$. The reduced length of $M$ follows as
\begin{multline}
		{\left\lfloor 4+v+\log_2(v) - w \right\rfloor} -\\
			(w = v-1) \cdot \max_{i \in \{ 0,\dots, \lfloor \log_2(v) \rfloor\}}
				\left( 2^i \mid v - 2^{\lfloor \log_2(v) \rfloor} \right) -\\
			(v \in 2^{\mathbb{N}_0}) \cdot \max_{i \in \{0,\dots,3\}}
				\left( 2^i \mid \lfloor \log_2(v) \rfloor \right),
\end{multline}
which was to be shown.\qed
\section{Proof of Proposition~\ref{prop:linear_takum-consecutive_integers}}\label{sec:proof-linear_takum-consecutive_integers}
Without loss of generality we can assume $m \in \mathbb{N}_1$, because the set of
$n$-bit takums is closed under negation and the integer zero is represented by
$\bm{0}_n$ for any $n$. Let us further assume that $m$ is an arbitrary $v$-bit integer
with $w$ trailing zeros in its binary representation, where $v \in \mathbb{N}_1$ and
$w \in \mathbb{N}_0$.
Our goal is to determine an upper bound for $v$ that depends on $n$ to find out the
largest consecutive integer represented by an $n$-bit takum.
\par
With Proposition~\ref{prop:takum-integer} we know there exists an $M \in {\{0,1\}}^\ell$ with $\pi(M)=m$ and $\ell$ as in (\ref{eq:prop-takum-integers-ell}). It holds that
\begin{align}
	\ell \le n
	&\Leftrightarrow {\left\lfloor 4+v+\log_2(v) - w \right\rfloor} -\notag\\
	&\hphantom{\Leftrightarrow}\quad (w \!= \!v\!-\!1) \cdot \max_{i \in \{ 0,\dots,\lfloor \log_2(q) \rfloor \}}
		\left( 2^i \mid v - 2^{\lfloor \log_2(q) \rfloor} \right) -\notag\\
	&\hphantom{\Leftrightarrow}\quad (v \in 2^{\mathbb{N}_0}) \cdot \max_{i \in \{0,\dots,3\}}
		\left( 2^i \mid \lfloor \log_2(v) \rfloor \right) \le n\label{eq:proof-takum-consecutive_integers-before}\\
	&\Leftarrow {\left\lfloor 4+q+\log_2(v) \right\rfloor} \le n\label{eq:proof-takum-consecutive_integers-after}\\
	&\Leftrightarrow 4+v+\log_2(v) < n + 1\\
	&\Leftrightarrow v+\log_2(v) < n - 3\\
	&\Leftrightarrow 2^{v+\log_2(v)} < 2^{n-3}\\
	&\Leftrightarrow v \cdot 2^v < 2^{n-3}\\
	&\Leftrightarrow \frac{\ln(2)}{\ln(2)} v \cdot \exp\!\left({\ln(2) v}\right) < 2^{n-3}\\
	&\Leftrightarrow \left( \ln(2)v \right) \exp\!\left({\ln(2) v}\right) < 2^{n-3} \ln(2).
\end{align}
If we set $\tilde{v} := \ln(2)v$ we obtain
\begin{equation}
	\ell \le n \Leftarrow \tilde{v} \cdot \exp(\tilde{v}) < 2^{n-3} \ln(2).
\end{equation}
We know that $v$ can be expressed using the \textsc{Lambert} $W$ function. Given
$2^{n-3} \ln(2) > 0$ we only need to consider its principal branch $W_0$. As $W_0$
is monotonically increasing it holds with $v \in \mathbb{N}_1$ (operation 1)
\begin{align}
	\tilde{v} \!\cdot\! \exp(\tilde{v}) \!<\! 2^{n-3} \ln(2) &\Leftrightarrow
	\tilde{v} < W_0\!\left(2^{n-3} \ln(2)\right)\\
	&\Leftrightarrow v < \frac{W_0\!\left(2^{n-3} \ln(2)\right)}{\ln(2)}\\
	&\overset{1}{\Leftrightarrow} v \le \left\lceil \frac{W_0\!\left(2^{n-3} \ln(2)\right)}{\ln(2)}
		\!-\! 1 \right\rceil.
\end{align}
Thus we have found an upper bound on $v$ such that arbitrary $v$-bit integers
are represented by an $n$-bit takum. Given the step from (\ref{eq:proof-takum-consecutive_integers-before}) to (\ref{eq:proof-takum-consecutive_integers-after}) this upper bound is not
tight and only a starting point for the next stage of the proof.
\par
Let $v= \left\lceil \frac{W_0\!\left(2^{n-3} \ln(2)\right)}{\ln(2)}
- 1 \right\rceil$ and $m=2^v-1$,
namely a $v$-bit saturated integer and thus the largest consecutive integer
value representable with an $n$-bit takum that we know of.
Now let us take a look at $m+1 = 2^v$, a $(v+1)$-bit integer with $v$ trailing
zeros in its binary representation. With Proposition~\ref{prop:takum-integer}
we know that there exists $M'\in {\{0,1\}}^{\ell'}$ with $\tau(M')=m+1$
and
\begin{align}
	\ell' &:= {\left\lfloor 4+(v\!+\!1)+\log_2(v\!+\!1) \!-\! v \right\rfloor} -
		\notag\\
	&\hphantom{\Leftrightarrow}\quad (v \!=\! (v\!+\!1)\!-\!1) \cdot\notag\\
	&\hphantom{\Leftrightarrow}\quad\max_{i \in
		\{ 0,\dots,\lfloor \log_2(v+1) \rfloor \}}
			\!\left( 2^i \mid (v+1) \!-\! 2^{\lfloor \log_2(v+1) \rfloor} \right) -\notag\\
	&\hphantom{\Leftrightarrow}\quad (v+1 \in 2^{\mathbb{N}_0}) \cdot \max_{i \in \{0,\dots,3\}}
			\left( 2^i \mid \lfloor \log_2(v+1) \rfloor \right)\\
	&\le {\left\lfloor 5 + \log_2(v + 1) \right\rfloor} -\notag\\
	&\hphantom{=}\,\,\,\,
		(v+1 \in 2^{\mathbb{N}_0}) \cdot \max_{i \in \{0,\dots,3\}}
		\left( 2^i \mid \lfloor \log_2(v+1) \rfloor \right)\\
	&= {\left\lfloor 5+\log_2\!\left(
		\left\lceil \frac{W_0\!\left(2^{n-3} \ln(2)\right)}{\ln(2)}
		- 1 \right\rceil+1
	\right) \right\rfloor} -\notag\\
	&\hphantom{=}\,\,\,\,
			(v+1 \in 2^{\mathbb{N}_0}) \cdot \max_{i \in \{0,\dots,3\}}
			\left( 2^i \mid \lfloor \log_2(v+1) \rfloor \right)\\
	&= {\left\lfloor 5+\log_2\!\left(
		\left\lceil \frac{W_0\!\left(2^{n-3} \ln(2)\right)}{\ln(2)}
		\right\rceil
	\right) \right\rfloor} -\notag\\
	&\hphantom{=}\,\,\,\,
			(v+1 \in 2^{\mathbb{N}_0}) \cdot \max_{i \in \{0,\dots,3\}}
			\left( 2^i \mid \lfloor \log_2(v+1) \rfloor \right)\\
	&\le {\left\lfloor 5+\log_2\!\left(
			1 + \frac{W_0\!\left(2^{n-3} \ln(2)\right)}{\ln(2)}
		\right) \right\rfloor} -\notag\\
	&\hphantom{=}\,\,\,\,
		(v+1 \in 2^{\mathbb{N}_0}) \cdot \max_{i \in \{0,\dots,3\}}
		\left( 2^i \mid \lfloor \log_2(v+1) \rfloor \right)\label{eq:proof-takum-consecutive_integers-ell_midresult}.
\end{align}
Using a result from Hoorfar et al. from \cite[Theorem~2.7]{lambert-inequality} that holds for $x > \mathrm{e}$ (operation 1)
we obtain the bound
\begin{align}
	W_0(x) &\overset{1}{\le} \ln\!\left( \frac{x}{\ln(x)} \right) + \frac{\mathrm{e}}{\mathrm{e} - 1}
		\frac{\ln(\ln(x))}{\ln(x)}\\
	&= \ln\!\left( \frac{x}{\ln(x)} \right) + \frac{\mathrm{e}}{(\mathrm{e} - 1) \ln(x)}
		\ln(\ln(x))\\
	&= \ln\!\left( \frac{x}{\ln(x)} \right) + \ln\!\left(
		{(\ln(x))}^{\frac{\mathrm{e}}{(\mathrm{e} - 1) \ln(x)}}
		\right)\\
	&= \ln\!\left(
			x \cdot {(\ln(x))}^{\frac{\mathrm{e}}{(\mathrm{e} - 1) \ln(x)} - 1}
		\right).
\end{align}
In our case $x := 2^{n-3} \ln(2) > \mathrm{e}$ holds for $n \in \mathbb{N}_5$, as
$2^2 \ln(2) \approx 2.77 > \mathrm{e}$. This means that we can use the bound (operation 1).
It also holds $\frac{\mathrm{e}}{(\mathrm{e} - 1) \ln(x)}-1 \in (0,1)$ (operation 2),
given $\ln(x) > 1$ is never negative and the expression is smaller than one
if and only if $\ln(x) > \frac{\mathrm{e}}{2(\mathrm{e}-1)} \approx 0.79$, which is the case
as for $n=5$ we know $\ln(x) = (5-3) \ln(2) + \ln(\ln(2)) \approx 1.02 > 0.79$ and
$\ln(x)$ is strictly monotonically increasing in $n$. It follows with
$\ln(\ln(2)) < 0$ (operation 3) that
\begin{align}
	W_0\!\left(2^{n-3} \ln(2)\right) &\overset{1}{\le}
		\ln\!\left(
			x \cdot {(\ln(x))}^{\frac{\mathrm{e}}{(\mathrm{e} - 1) \ln(x)} - 1}
		\right)\\
	&\overset{2}{<} \ln\!\left(
			x \cdot \ln\!\left(x\right)
		\right)\\
	&= \ln\!\left(
			2^{n-3} \ln(2) \cdot\right.\notag\\
	&\hphantom{=}\,\, \left. \left[ (n-3) \ln(2) + \ln(\ln(2)) \right] \right)\\
	&\overset{3}{<} \ln\!\left(
				2^{n-3} \cdot \ln(2) \cdot (n-3) \cdot \ln(2) \right)\\
	&= (n-3) \ln(2) + \ln(n-3) +\notag\\
	&\hphantom{=}\,\,\, 2\ln(\ln(2)).\label{eq:proof-takum-consecutive_integers-W0_bound}
\end{align}
We insert (\ref{eq:proof-takum-consecutive_integers-W0_bound}) into
(\ref{eq:proof-takum-consecutive_integers-ell_midresult}) and obtain
{\allowdisplaybreaks
\begin{align}
	\ell' &\le {\left\lfloor 5\!+\!\log_2\!\left(
		1 \!+\! \frac{ (n\!-\!3)\ln(2) \!+\! \ln(n\!-\!3) \!+\!
			2\ln(\ln(2)) }{\ln(2)}
		\right) \right\rfloor} -\notag\\
	&\hphantom{=}\,\,\,\,
		(v+1 \in 2^{\mathbb{N}_0}) \cdot \max_{i \in \{0,\dots,3\}}
		\left( 2^i \mid \lfloor \log_2(v+1) \rfloor \right)\\
	&= {\left\lfloor 5+\log_2\!\left(
		n - 2 + \frac{ \ln(n-3) + 2\ln(\ln(2)) }{\ln(2)}
		\right) \right\rfloor}-\notag\\
	&\hphantom{=}\,\,\,\,
		(v+1 \in 2^{\mathbb{N}_0}) \cdot \max_{i \in \{0,\dots,3\}}
		\left( 2^i \mid \lfloor \log_2(v+1) \rfloor \right)\label{eq:proof-takum-consecutive_integers-last_bound}\\
	&\begin{cases}
		= 6 - 1 = 5 & n = 5\\
		= 7 - 1 = 6 & n = 6\\
		= 7 & n = 7\\
		< n & n \ge 8
	\end{cases}\\
	&\le n,
\end{align}
}
where the last case is handled by noting that
(\ref{eq:proof-takum-consecutive_integers-last_bound}) is at most
$7$ for $n=8$ and the first summand grows in the order of
$\mathcal{O}(\ln(n))$.
As we can see $m+1$ has a takum representation that fits within $n$
bits, which means that it is also part of the set of representable
consecutive integers.
\par
Let us now take a look at $m+2 = 2^v+1$, a $v+1$-bit integer with
zero trailing zeros in its binary representation ($w=0$). With
Proposition~\ref{prop:takum-integer} we know that there exists
$M'' \in {\{0,1\}}^{\ell''}$ with $\tau(M'')=m+2$ and, with
$\mathbb{N}_1 \ni v \neq 0$ (operation 1),
\begin{align}
	\ell'' &:= {\left\lfloor 4+(v\!+\!1)+\log_2(v+1) \right\rfloor} -
			\notag\\
	&\hphantom{\Leftrightarrow}\quad (0 = (v+1)-1) \cdot\notag\\
	&\hphantom{\Leftrightarrow}\quad \max_{i \in \{ 0,\dots,\lfloor \log_2(v+1) \rfloor \}}
				\!\left( 2^i \mid (v+1) \!-\! 2^{\lfloor \log_2(v+1) \rfloor} \right) -\notag\\
	&\hphantom{\Leftrightarrow}\quad (v+1 \in 2^{\mathbb{N}_0}) \cdot \max_{i \in \{0,\dots,3\}}
				\left( 2^i \mid \lfloor \log_2(v+1) \rfloor \right)\\
	&\overset{1}{=}
		{\left\lfloor 4+(v+1)+\log_2(v+1) \right\rfloor}\\
	&= {\left\lfloor 5+v+\log_2(v+1) \right\rfloor}\\
	&> 4 + v + \log_2(v+1)\\
	&= 4 + \left\lceil \frac{W_0\!\left(2^{n-3} \ln(2)\right)}{\ln(2)}
	- 1 \right\rceil + \notag\\
	&\hphantom{=}\,\,\log_2\!\left(
			\left\lceil \frac{W_0\!\left(2^{n-3} \ln(2)\right)}{\ln(2)}
			\right\rceil
		\right)\\
	&\ge 3 + \frac{W_0\!\left(2^{n-3} \ln(2)\right)}{\ln(2)} +
		\log_2\!\left(
			\frac{W_0\!\left(2^{n-3} \ln(2)\right)}{\ln(2)}
		\right)\\
	&= 3 + \frac{\ln(2^{n-3})}{\ln(2)}\\
	&= 3 + n - 3\\
	&= n.
\end{align}
Here we can see that $m+2$'s posit representation does not
fit within $n$ bits. Thus $m+1=2^v$ is the largest representable
consecutive integer, as was to be shown.\qed
\printbibliography
\end{document}